\documentclass[letterpaper, journal, 11pt, final]{article}

\usepackage[T1]{fontenc} \usepackage{amsmath,graphicx}
\usepackage{subfigure} \usepackage{epstopdf} \usepackage{algorithm}
\usepackage{algpseudocode} \usepackage{color}
\usepackage{citesort}

 \newtheorem{prop}{Proposition}
\newtheorem{corol}{Corrolary}
\newenvironment{proof}[1][Proof~~--]{\begin{trivlist}
  \item[\hskip \labelsep {\bfseries #1}]}{\end{trivlist}}

\input alphabet.tex % Definitions de macros pour les alphabeiiits
\input abrege.tex % Abbreviations latines, anglaises et fran\;c}aises
\def\cro#1{\left[#1\right]}

\def\Exp#1{\exp\cro{#1}}

\newsavebox{\fminibox}
\newlength{\fminilength}
\newenvironment{fminipage}[1][\linewidth]
  {\setlength{\fminilength}{#1}%-2\fboxsep-2\fboxrule}%
   \begin{lrbox}{\fminibox}\begin{minipage}{\fminilength}}
  {\end{minipage}\end{lrbox}\noindent\fbox{\usebox{\fminibox}}}

 \def\T{^\tD} \def\+{^\dagger}
\def\nequiv{\not\kern-.05em\equiv}
\def\egal{\kern-.5em=\kern-.5em}        % Moins d'espace autour de "="
\def\propt{\kern-.2em\propto\kern-.2em} % Idem
\def\intdouble{\int\kern-0.3em\int}
\def\inttriple{\int\kern-0.3em\int\kern-0.3em\int}

\def\rond#1{\overset{\kern-0.33em~_\circ}{#1}}
\def\rondit[#1]#2{\overset{\kern#1~_\circ}{#2}}

\def\babs{\begin{abstract}}             \def\eabs{\end{abstract}}
\def\barr{\begin{array}}                \def\earr{\end{array}}
\def\bcc{\begin{center}}                \def\ecc{\end{center}}

\def\bdes{\begin{description}}          \def\edes{\end{description}}
\def\bdoc{\begin{document}}             \def\edoc{\end{document}}
\def\ben{\begin{enumerate}}             \def\een{\end{enumerate}}
\def\beqn{\begin{eqnarray}}             \def\eeqn{\end{eqnarray}}
\def\beqnl#1{\beqn\label{#1}}           \def\eeqnl#1{\label{#1}\eeqn}
\def\beqnx{\begin{eqnarray*}}           \def\eeqnx{\end{eqnarray*}}
\def\bseqn{\begin{subeqnarray}}         \def\eseqn{\end{subeqnarray}}
\def\beq#1\eeq{\begin{equation}#1\end{equation}}
\def\bal#1\eal{\begin{align}#1\end{align}}
\def\balx#1\ealx{\begin{align*}#1\end{align*}}
\def\beqx{$$}                           \def\eeqx{$$}
\def\bfig{\protect\begin{figure}}       \def\efig{\protect\end{figure}}
\def\bfigx{\protect\begin{figure*}}     \def\efigx{\protect\end{figure*}}
\def\bfigt{\protect\begin{figurette}}   \def\efigt{\protect\end{figurette}}
\def\bfl{\begin{flushleft}}             \def\efl{\end{flushleft}}
\def\bfr{\begin{flushright}}            \def\efr{\end{flushright}}
\def\bit{\begin{itemize}}               \def\eit{\end{itemize}}
\def\bmi{\begin{minipage}}              \def\emi{\end{minipage}}
\def\bfmi{\begin{fminipage}}            \def\efmi{\end{fminipage}}
\def\bpic{\begin{picture}}              \def\epic{\end{picture}}
\def\bqu{\begin{quote}}                 \def\equ{\end{quote}}
\def\bqun{\begin{quotation}}            \def\equn{\end{quotation}}
\def\bsl{\begin{slide}}                 \def\esl{\end{slide}}
\def\btabb{\begin{tabbing}}             \def\etabb{\end{tabbing}}
\def\btabl{\begin{table}}               \def\etabl{\end{table}}
\def\btablx{\begin{table*}}             \def\etablx{\end{table*}}
\def\btab{\begin{tabular}} %\def\etab{\end{tabular}} Conflit avec eta gras... Elle est pas grasse, ma soeur !
\def\btabu{\begin{tabular}}             \def\etabu{\end{tabular}}
\def\btabx{\begin{tabular*}}            \def\etabx{\end{tabular*}}
\def\bbib{\begin{thebibliography}{999}} \def\ebib{\end{thebibliography}}
\def\bver{\begin{verbatim}}             \def\ever{\end{verbatim}}
\def\bca{\begin{cases}}                  \def\eca{\end{cases}}
\def\Gn{\gamma_n} \def\Gx{\gamma_x} \def\xbh{\widehat \xb}
\def\epsb{{\ensuremath{\varepsilonb}}}
\def\Rbmun{{\ensuremath{\Rb_n^{-1}}}} \def\Rx{{\ensuremath{\Rb_x}}}
\def\Rxmun{{\ensuremath{\Rb_x^{-1}}}}
\def\Rxapost{{\ensuremath{\Rb_x^{\rm post}}}}
\def\mx{{\ensuremath{\mb_x}}}
\def\mxapost{{\ensuremath{\mb_x^{\rm post}}}}
\def\xbt{{\ensuremath{\widetilde{\xb}}}}
\def\mubt{{\ensuremath{\widetilde{\mub}}}}

\definecolor{cacadoa}{rgb}{0.5,0.5,0} \def\Vert#1{{\color{green} #1}}
\def\FORQ#1{{\color{blue} #1}} \def\Joli#1{{\color{cacadoa} #1}}
\def\Joli#1{{\color{magenta} #1}}
\def\Gio#1{{\color{red}#1}}
\def\GioSL#1{\textsl{\small\color{red} (#1)}}
\newcommand{\KK}[1]{{\st{#1}}}
\begin{document}

\title{Efficient sampling of high-dimensional Gaussian fields: the
  non-stationary / non-sparse case}

\author{F. Orieux$^{*}$, O. F\'eron and J.-F. Giovannelli\thanks{F.
    Orieux is with Pasteur Institute, 25 rue du Dr Roux, 75015
    Paris, France, \texttt{orieux@pasteur.fr}. O. F\'eron is with {EDF Research \&
      Developments}, Dpt OSIRIS, 92140 Clamart, France,
    \texttt{olivier-2.feron@edf.fr}. J.-F. Giovannelli is with the
    Laboratoire de l'Int\'egration du Mat\'eriau au Syst\`eme, 33405
    Talence, France, \texttt{Giova@IMS-Bordeaux.fr}.}}

\maketitle

\begin{abstract}
  This paper is devoted to the problem of sampling Gaussian fields in high dimension. Solutions exist for two specific structures of inverse covariance : sparse and circulant. The proposed approach is valid in a more general case and especially as it emerges in inverse problems. It relies on a perturbation-optimization principle: adequate stochastic perturbation of a criterion and optimization of the perturbed criterion. It is shown that the criterion minimizer is a sample of the target density. The motivation in inverse problems is related to general (non-convolutive) linear observation models and their resolution in a Bayesian framework implemented through sampling algorithms when existing samplers are not feasible. It finds a direct application in myopic and/or unsupervised inversion as well as in some non-Gaussian inversion. An illustration focused on hyperparameter estimation for super-resolution problems assesses the effectiveness of the proposed approach.
\end{abstract}

% \begin{IEEEkeywords}
%   Stochastic sampling, high-dimensional sampling, inverse problem,
%   Bayesian strategy, unsupervised, myopic% (semi-blind)
% \end{IEEEkeywords}

% \begin{center}
%   EDICS: IMD-MDSP
% \end{center}

\section{Introduction}\label{sec_intro}

This work deals with simulation of high-dimensional Gaussian and
conditional Gaussian fields. The problem difficulty is directly
related to handling high-dimensional covariances $\Rb$ and precision
matrices $\Qb=\Rb^{-1}$.  In\-ver\-sion and factorization of these
matrices can be very costly in terms of time and memory, if not
impossible. General tools~\cite{Robert04,Gilks96} provide
pixel-by-pixel sequential Gibbs or Hastings-Metropolis algorithms but
they are not practicable in high dimension. This problem is old and
solutions exist in two cases.
%
%
% \cite{Skilling06}\cite{Geman84}Robert07
%
%
\begin{itemize}
% \item[1~--] \FORQ{When $\Qb$ is sparse, solutions exist to avoid its
%     inversion. The most common methods are the chessboard-like
%     decomposition \cite{Feron07} used in a parallel Gibbs sampler or
%     fast sampling methods \cite{Rue01} with a Cholesky decomposition
%     of $\Qb=\Lb\Lb\T$. Then, as \cite{Rue01,Lalanne01} pointed out,
%     a sample can be obtained by solving $\Lb \xb = \epsb$, with
%     $\epsb$ a white Gaussian noise. The sparsity of $\Lb$ ensures
%     feasible numerical resolution of the linear system.}

%   For instance, Gaussian Markov Random fields (GMRF) are very
%   popular in image processing and allow to model correlated Gaussian
%   fields from a sparse precision matrix, thanks to the Markov
%   property.

\item When $\Qb$ is sparse, two strategies have been proposed. The
  first one~\cite[chap. 8]{Winkler03}, relies on a parallel Gibbs
  sampler based on a chessboard-like decomposition. It takes advantage
  of the sparsity of $\Qb$ to allow large blocks of variables to be
  simultaneously updated. The second strategy~\cite{Rue01,Lalanne01}
  relies on a Cholesky decomposition $\Qb=\Lb\T\Lb$: a sample $\xb$
  can be obtained by solving the linear system $\Lb \xb = \epsb$,
  where $\epsb$ is a zero-mean white Gaussian vector. The sparsity of
  $\Lb$ ensures feasible numerical resolution of the linear system.

  % Unfortunately, when $\Qb$ is not sparse, these methods are not
  % applicable since the factorization of $\Qb$ or the direct
  % resolution of problem $\Lb \xb = \epsb$ become intractable.

\item In \cite{Chellappa85,Chellappa92} the authors pointed out an
  efficient solution for the case of circulant matrix $\Qb$, even
  non-sparse. In this case, the covariance is diagonal in the Fourier
  domain: the sampling is based on independent sampling of the Fourier
  coefficients. Finally, the sampling is efficiently computed by FFT
  and it has been used in~\cite{Geman95,Giovannelli08,Orieux10}.
\end{itemize}
To our knowledge there is no solution for sampling more general
high-dimensional Gaussian fields. In this paper we propose an
efficient algorithm for a more general case where $\Qb$ is non-sparse,
non-circulant and very large.  The proposed approach is applicable to
any precision matrix of the form
\begin{equation}
  \label{eq_precision_form}
  \Qb=\sum_{k=1}^K \Mb_k\T\Rb_k^{-1} \Mb_k
\end{equation}
for which, to the best of our knowledge, no practical solution exists.
A recent paper~\cite{Tan10} briefly describes a similar algorithm
  for a compress sensing problem in signal processing. Our paper deepens
and generalizes this contribution.

The problem of sampling such fields is commonly encountered in
Bayesian approaches for inverse problems and especially in high
dimension like in image reconstruction. Indeed, let us consider the
general linear forward model
\begin{equation}
  \label{eq_mod_dir}
  \yb=\Hb \xb + \nb,
\end{equation}
where $\yb$, $\nb$ and $\xb$ denote the observations, the noise and
the unknown image and $\Hb$ is a linear operator. Consider, again, two
prior densities for $\nb$ and $\xb$ that are Gaussian conditionally to
a set of parameters $\thetab$ and focus on the joint estimation of
$\xb$ and $\thetab$ from the posterior density $p(\xb,\thetab|\yb)$.
This framework is very general and can be used in many applications.
In image reconstruction, it covers a majority of current problems such
as unsupervised \cite{Giovannelli08} or myopic \cite{Orieux10}
inversion, since acquisition (or instrument) parameters and
hyperparameters can be included in $\thetab$. The framework also
covers some non-Gaussian priors involving auxiliary/hidden
variables~\cite{Geman84,Geman95,Giovannelli08,Feron07,Ayasso10}
(location mixture or scale mixture of Gaussian), by including these
variables in $\thetab$.

The joint estimation of $\xb$ and $\thetab$ from the posterior density
$p(\xb,\thetab|\yb)$ commonly requires the handling of the posterior
conditional probability $p(\xb|\thetab,\yb)$. Under the above
assumption, this density is Gaussian with precision matrix $\Qb$ of
the form \eqref{eq_precision_form}, as shown in
section~\ref{sec_prob_inv}. The capability to sample from this density
makes it possible to propose, for instance, stochastic optimization
\cite{Geman84} or Gibbs sampler \cite{Feron07,Orieux10}. In the
general case of inverse problems, $\Qb$ is neither sparse nor
circulant so existing sampling methods fail whereas the proposed
sampling method is effective.

% to any inverse problems related to precision matrix of the form
% \eqref{eq_precision_form}. %It then offers the possibility to
% propose efficient algorithms for unsupervised and myopic inversion.

% This consists in solving an optimization problem where the criterion
% is perturbed by Gaussian noises. The development is actually founded
% on the fact that a particular factorization of the precision matrix,
% described by \eqref{eq_precision_form}, is available. A sample
% cannot be done by solving a linear sparse system, as in
% \cite{Rue01,lalanne01} but is obtained by using an optimization
% procedure. Thus, the proposed algorithm solves the problems of
% sampling Gaussian fields in the case that did not yet have, to our
% knowledge, a solution: when the operators are neither sparse nor
% circulant and the dimension is very high.

Subsequently, section~\ref{sec_pert_opt} presents the proposed
algorithm and its direct application to general inverse
problems. Section~\ref{sec_illustr} illustrates the algorithm through
an academic inverse problem in super-resolution imaging.  Section
\ref{sec_conclu} concludes and presents some perspectives.

\section{Perturbation-optimization algorithm}
\label{sec_pert_opt}

\subsection{Description}

Here we focus on the problem of sampling from a target Gaussian
density $\Nc(0,\Qb^{-1})$ where $\Qb$ is in the form
\eqref{eq_precision_form}. When $\Qb$ is neither sparse nor circulant,
existing methods fail in high dimension and we propose an efficient
solution based on the Perturbation-Optimization (PO) algorithm
described by Algorithm~\ref{algo_PO} and
Proposition~\ref{proposition_PO}.

\begin{algorithm}[htbp]
  \caption{: Perturbation-Optimization algorithm.\label{algo_PO}}
  \begin{algorithmic}[1]
    \State \textbf{Step P (Perturbation)}: Generate independent
    Gaussian variables $\etab_k,~k=1,\dots,K$ following
    \begin{equation}
      \label{eq_loi_eta}
      \etab_k \sim \Nc(0,\Rb_k), \quad \forall k=1,\dots K
    \end{equation}
    \State \textbf{Step O (Optimization)}: Compute $\xbh$ as the
    minimizer of the criterion
    \begin{equation}\label{eq_critere_general}
      J(\xb|\etab_1,\dots,\etab_K)=\sum_{k}^K \left(\etab_k - \Mb_k \xb
      \right)\T \Rb_k^{-1}\left(\etab_k - \Mb_k \xb \right)
    \end{equation}
  \end{algorithmic}
\end{algorithm}

\begin{prop}
  \label{proposition_PO}
  The minimizer $\xbh$ of criterion~\eqref{eq_critere_general}
  resulting from Algorithm~\ref{algo_PO} is Gaussian
  \begin{equation}
    \label{eq:6}
    \xbh \sim \Nc(0,\Qb^{-1})\,.
  \end{equation}
\end{prop}
%\medskip

\begin{proof}
  The minimizer $\xbh$ of criterion~\eqref{eq_critere_general} has an
  analytical expression:
  \begin{equation}
    \label{eq_sol_xbh}
    \begin{split}
      \xbh & = \left[\sum_{k=1}^K \Mb_k\T\Rb_k^{-1} \Mb_k\right]^{-1}
      \left( \sum_{k=1}^K \Mb_k\T\Rb_k^{-1} \etab_k \right) \\
      & = \Qb^{-1}\left( \sum_{k=1}^K \Mb_k\T\Rb_k^{-1} \etab_k
      \right) \,.
    \end{split}
  \end{equation}
  It is clearly a zero-mean Gaussian vector as a linear combination of
  $K$ zero-mean Gaussian vectors. The covariance is calculated below
  using elementary algebra: from \eqref{eq_loi_eta} and
  \eqref{eq_sol_xbh}, we have
  \begin{align*}
    \Vbb [\xbh] %& = \Ebb \left[\left(\xbh \xbh\T \right)\right] \\
    & = \Qb^{-1} \Big[ \sum_{k,k'=1}^K \Mb_k\T \Rb_k^{-1}
    \Ebb\left[\etab_k \etab_{k'}\T\right] \Rb_{k'}^{-1} \Mb_{k'} \Big]
    \Qb^{-1} \\
    & = \Qb^{-1} \Big[ \sum_{k=1}^K \Mb_k\T \Rb_k^{-1}
    \Ebb\left[\etab_k \etab_k\T\right] \Rb_k^{-1} \Mb_k \Big] \Qb^{-1} \\
    & = \Qb^{-1} \Big[ \sum_{k=1}^K \Mb_k\T \Rb_k^{-1}  \Mb_k \Big] \Qb^{-1} = \Qb^{-1}%\\
    %& = \Qb^{-1}
  \end{align*}
  that completes the proof.
\end{proof}
%\medskip

% From the previous result, we deduce an algorithm for sampling a
% target Gaussian fields $\Nc(0,\Qb^{-1})$, with $\Qb$ of the form
% \eqref{eq_precision_form} which does not require the inversion of
% $\Qb$. This algorithm is done in two steps.
%\begin{itemize}
%\item \textbf{Step P (perturbation)}: this step consists in generating
%  perturbations $\etab_k$ following \eqref{eq_loi_eta}.
%\item \textbf{Step O (optimization)}: this step consists in minimizing
%  the perturbed criterion $J(\xb|\etab_1,\dots,\etab_K)$.
%\end{itemize}

The criterion $J(\xb|\etab_1,\dots,\etab_K)$ being quadratic, we have
access to the whole available literature on efficient numerical
optimization tools, e.g. iterative techniques such as gradient based
ones (standard, corrected, conjugate, optimal step size\dots). We have
to highlight that in theory the sample of the target density is the
exact optimum of the perturbed criterion. Therefore the optimization
step may require as much descent steps as the dimension of the
problem. However, the optimization procedure can be stopped more
rapidly without practical loss of efficiency.

Obviously, the efficiency of the algorithm depends on the capability
to easily sample from Gaussian densities $\Nc(0,\Rb_k)$. This will be
actually the case in inverse problem applications as shown in
section~\ref{sec_prob_inv}.

\begin{algorithm}[htbp]
  \caption{: Perturbation-Optimization algorithm.}
  \label{algo_PO_mean}
  \begin{algorithmic}[1]
    \State \textbf{Step P (Perturbation)}: Generate independent
    Gaussian variables $\zetab_k,~k=1,\dots,K$ following
    \begin{equation}
      \label{eq_loi_mu}
      \zetab_k \sim \Nc(\mb_k,\Rb_k), \quad \forall k=1,\dots K
    \end{equation}
    \State \textbf{Step O (Optimization)}: Compute $\xbt$ as the
    minimizer of the criterion
    \begin{equation*}\label{eq_critere_mean}
      J(\xb|\zetab_1,\dots,\zetab_K)=\sum_{k=1}^K \left(\zetab_k - \Mb_k
        \xb \right)\T \Rb_k^{-1}\left(\zetab_k - \Mb_k \xb \right)
    \end{equation*}
  \end{algorithmic}
\end{algorithm}

Moreover, we can actually extend Proposition~\ref{proposition_PO} and
Algorithm~\ref{algo_PO} when the mean of the target Gaussian density
is not zero, by proposing Algorithm~\ref{algo_PO_mean} above and
Corrolary~\ref{corrolaire_PO_mean} below.

\begin{corol}
  \label{corrolaire_PO_mean}
  The solution $\xbt$ resulting from Algorithm~\ref{algo_PO_mean} is
  Gaussian
  \begin{equation}
    \label{eq_cond_gauss}
    \xbt \sim \Nc\left(\Qb^{-1} \left( \sum_{k=1}^K \Mb_k\T
        \Rb_k^{-1}\mb_k \right),\Qb^{-1}  \right)\;.
  \end{equation}
\end{corol}

\begin{proof}
  Consider $\etab_k = \zetab_k-\mb_k$, $k=1,\dots,K$, and the
  minimizer $\xbh$ of the criterion \eqref{eq_critere_general}. Hence
  it is trivial to show that $\xbt = \xbh + \Qb^{-1}\left( \sum_{k=1}^K
    \Mb_k\T \Rb_k^{-1}\mb_k \right) $.  Using the results of
  Proposition~\ref{proposition_PO} on $\xbh$, we can show
  \begin{align*}
    \Ebb\left[\xbt \right] & = \Qb^{-1} \left( \sum_{k=1}^K \Mb_k\T \Rb_k^{-1}\mb_k \right)  \\
    \Vbb \left[ \xbt\right] & = \Vbb \left[ \xbh\right] = \Qb^{-1}
  \end{align*}
  and that completes the proof.
\end{proof}

% In section~\ref{sec_prob_inv} we show that Gaussian density of the
% form \eqref{eq_cond_gauss} commonly appears on unsupervised and/or
% myopic inversion in a Bayesian framework.

\subsection{Application to inverse problems}\label{sec_prob_inv}

The purpose is to solve an inverse problem, stated by the forward
model~\eqref{eq_mod_dir}, in a Bayesian framework based on the
following models:
\begin{itemize}
\item $\Hb$ describes an observation system that can depend on unknown
  acquisition parameters,

\item prior densities for the observation noise and for the object are
  Gaussian $\nb \sim \Nc(\mb_n,\Rb_n)$ and $\xb \sim \Nc(\mx,\Rx)$,
  conditionally on a set of auxiliary variables.

  % \item observation noise is Gaussian $\nb \sim \Nc(\mb_n,\Rb_n)$,
  %   conditionally on a set of auxiliary variables.
  % \item prior density for the object is Gaussian $\xb \sim
  %   \Nc(\mx,\Rx)$ conditionally on a set of auxiliary variables.
\end{itemize}
In a general statement, $\thetab$ collects acquisition parameters,
hyperparameters and auxiliary variables. This framework covers myopic
(semi-blind) and unsupervised inversion, non-stationary or
inhomogeneous Gaussian priors and non-Gaussian priors involving
auxiliary variables.

The general inversion problem then consists in estimating $\xb$ and
$\thetab$ through the density $p(\xb, \thetab|\yb)$. The posterior
mean can be approximated using a Gibbs sampler. It is
an iterative algorithm which alternately samples from
$p(\thetab|\xb,\yb)$ and $p(\xb|\thetab,\yb)$. The conditional
posterior $p(\xb|\yb,\thetab)$ is a correlated Gaussian field: $\xb
\sim \Nc (\mxapost,\Rxapost)$ with
\begin{align*}
  \Rxapost & = \left(\Hb^t \Rbmun \Hb + \Rxmun \right)^{-1} \\
  \mxapost & = \Rxapost \left( \Hb^t \Rbmun \left[\yb-\mb_n \right] +
    \Rxmun \mx \right)
\end{align*}
where $\thetab$ is embedded in $\Hb, \Rb_n$ and $\Rx$ for simpler
notations.

If $\Hb$ has no particular properties then the precision matrix
$\Qb=(\Rxapost)^{-1}$ is neither sparse nor circulant, and existing
sampling methods are not applicable. The Perturbation-Optimization
algorithm makes it possible to efficiently sample from
$\Nc(\mxapost,\Rxapost)$. In particular, applying
Algorithm~\ref{algo_PO_mean} with $K=2$, $\Mb_1=\Hb$, $\Mb_2=\Ib$,
$\Rb_1=\Rb_n$, $\Rb_2=\Rx$, $\mb_1=\mb_n$ and $\mb_2=\mb_x$, directly
gives a sample from this density. Then, this algorithm ensures that
correct posterior mean and covariance are obtained, at the same
time. This increases the usefulness of this method for inverse
problems.

% As highlighted in section~\ref{sec_pert_opt} the efficiency of the
% Perturbation-Optimization algorithm depends on the capability to
% easily obtain samples from $\Nc(0,\Rb_n)$ and $\Nc(0,\Rx)$. In
% practice both covariance matrices are sparse.

\section{Illustration}
\label{sec_illustr}

The proposed PO algorithm is an effective tool for high dimensional
inverse problems, e.g. image reconstruction. In this context, it opens
up the possibility to resort to stochastic sampling algorithms (MCMC,
Gibbs, Metropolis-Hastings,\dots) providing two main advantages:
\begin{itemize}
\item the capability to jointly estimate several unknowns when the
  global modelization is more natural through conditional
  distributions (hierarchical structure),

\item in addition, the access to the entire distribution of the
  unknowns providing uncertainties (standard deviations, confidence
  intervals,\dots).
\end{itemize}
\vspace{-0.5cm}

\subsection{Two examples: electromagnetics and fluorescent microscopy}

For example, the proposed PO algorithm has been applied to an
electromagnetic inverse scattering problem by one of the
authors~\cite{Feron07}. In a domain integral representation, the
forward model expresses observed data as a bi-linear function of
unknown object and unknown induced current. The bi-linear structure
leads to a modelization with conditional Gaussians: the prior for the
induced current is Gaussian given the object, and the prior for the
object is Gaussian given the induced current. The joint estimation of
the object and the current is tackled in a Bayesian framework and
computed by means of a Gibbs sampler in which the sampling of the
current is made possible thanks to the proposed PO algorithm.

In \cite{orieux11} it has been applied by another one of the authors
to process data in biology imaging to achieve super-resolution in
fluorescent microscopy trough Structured Illumination. The problem is
also tackled in a Bayesian framework and implemented by means of a
Gibbs sampler. The density for the object given the other variables is
Gaussian with non-invariant covariance (due to non-invariant
illumination of the biological sample) making the use of existing
techniques impossible. Again, the proposed PO algorithm overcomes this
difficulty and results in the capability to estimate hyperparameters
and acquisition parameters, while also providing uncertainties.

\subsection{Unsupervised super-resolution}

In the following, we detail an application of the proposed PO
algorithm to the super-resolution (SR) academic problem: several
blurred and down-sampled (low resolution) images of a scene are
available in order to retrieve the original (high resolution) scene
\cite{Park03,Rochefort06}. It is shown that the crucial novelty,
enabled by the proposed PO algorithm, is to allow the use of sampling
algorithms in SR methods and to provide joint image and
hyperparameters estimation including uncertainties.

\begin{figure}[htbp]
  \centering

  \subfigure[$\Gn$ chain]{%
    \includegraphics[width=0.45\textwidth]{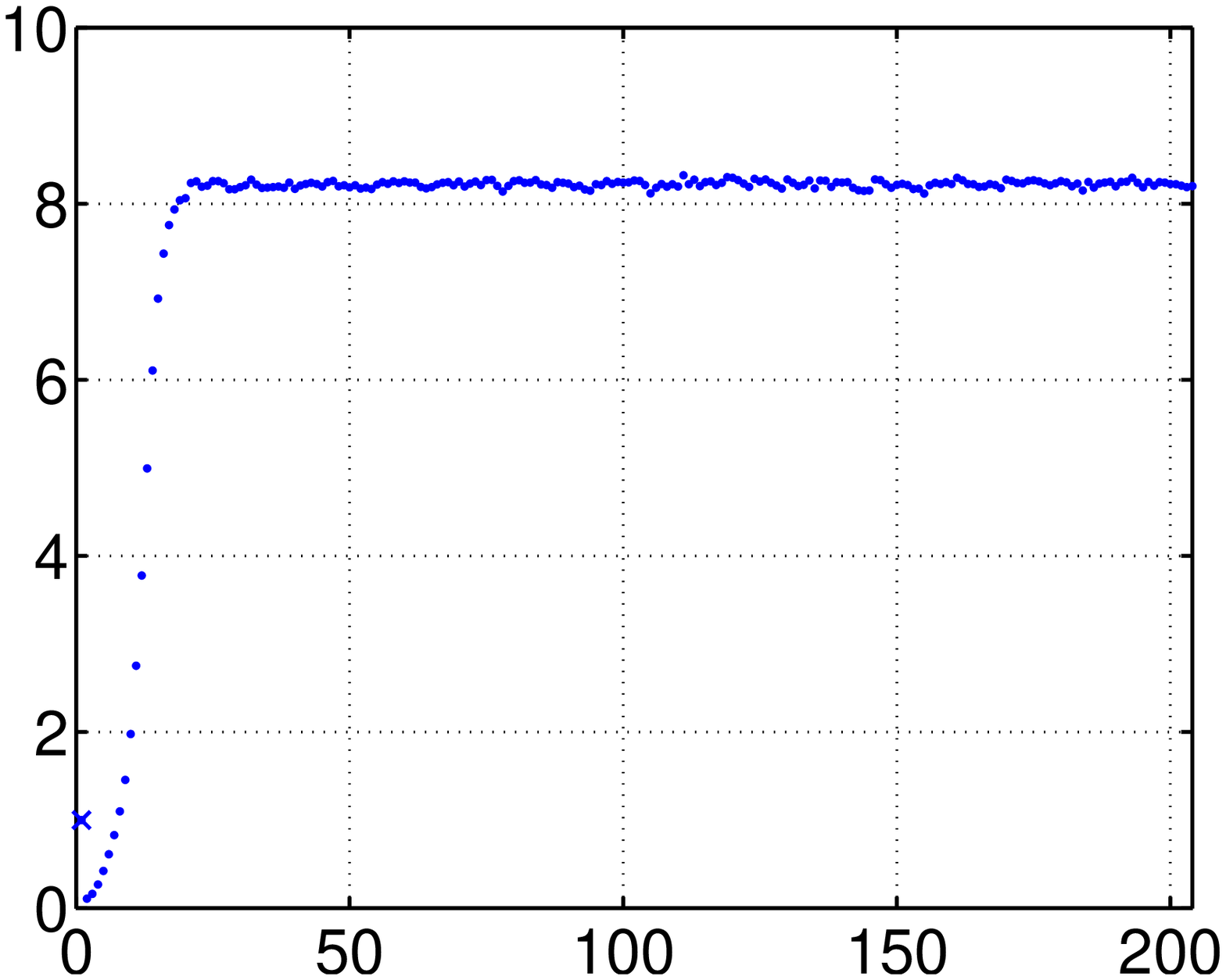}\label{fig:gNchain}}%
  ~~~~~%\hfill%
  \subfigure[$\Gn$ histogram]{%
    \includegraphics[width=0.45\textwidth]{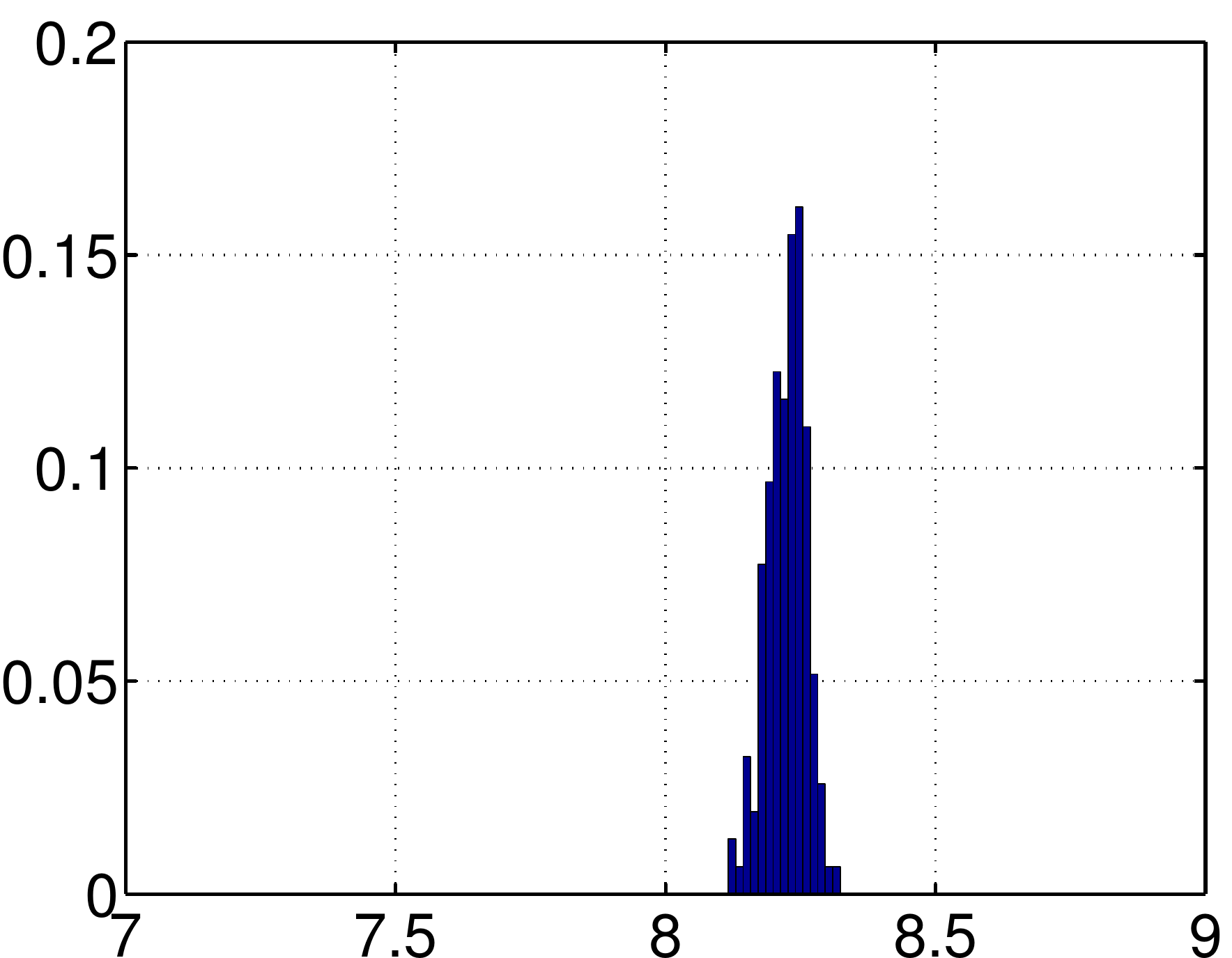}\label{fig:gNhist}}

  \subfigure[$\Gx$ chain]{%
    \includegraphics[width=0.45\textwidth]{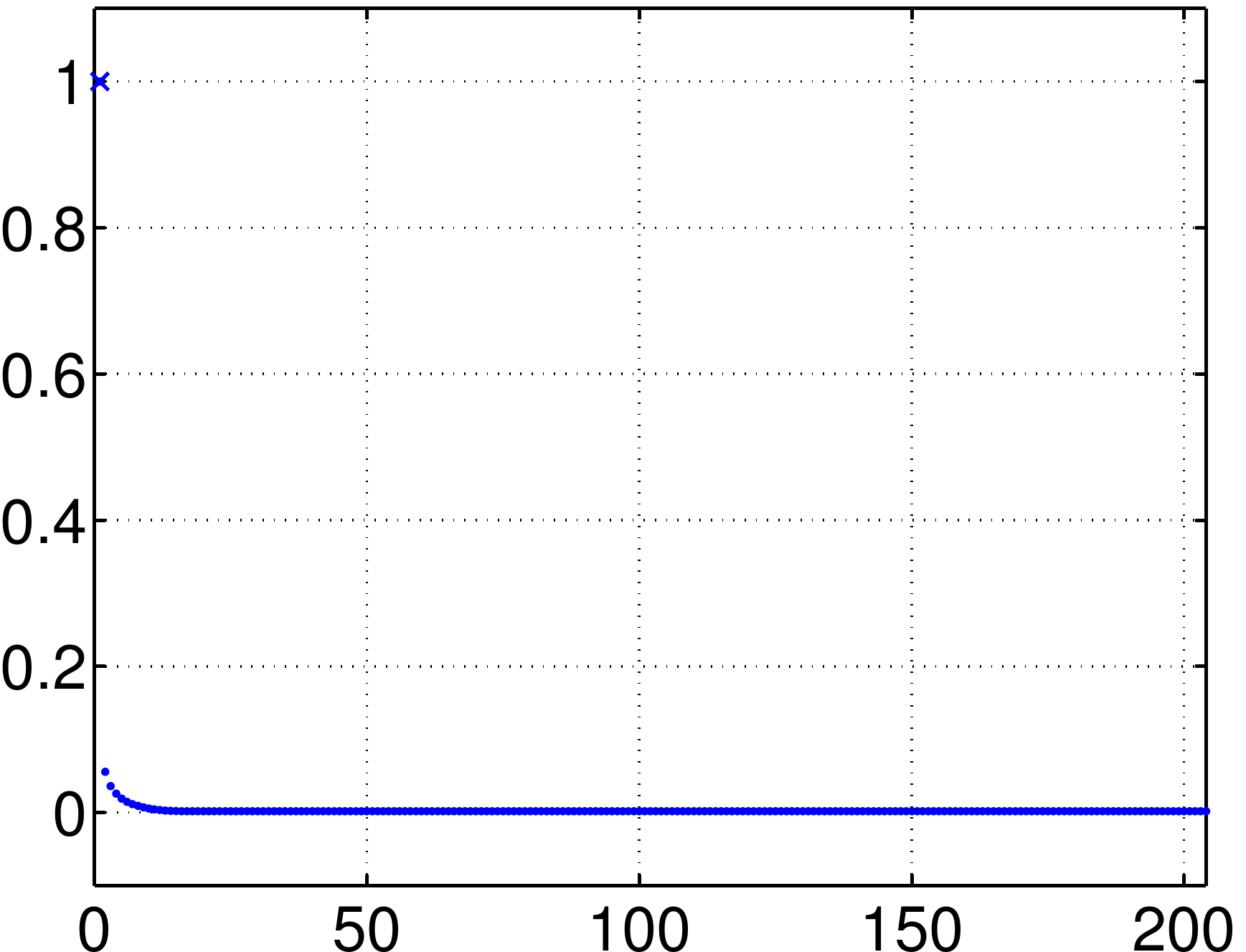} \label{fig:gXchain}}%
  ~~~~~%\hfill%
  \subfigure[$\Gx$ histogram]{%
    \includegraphics[width=0.45\textwidth]{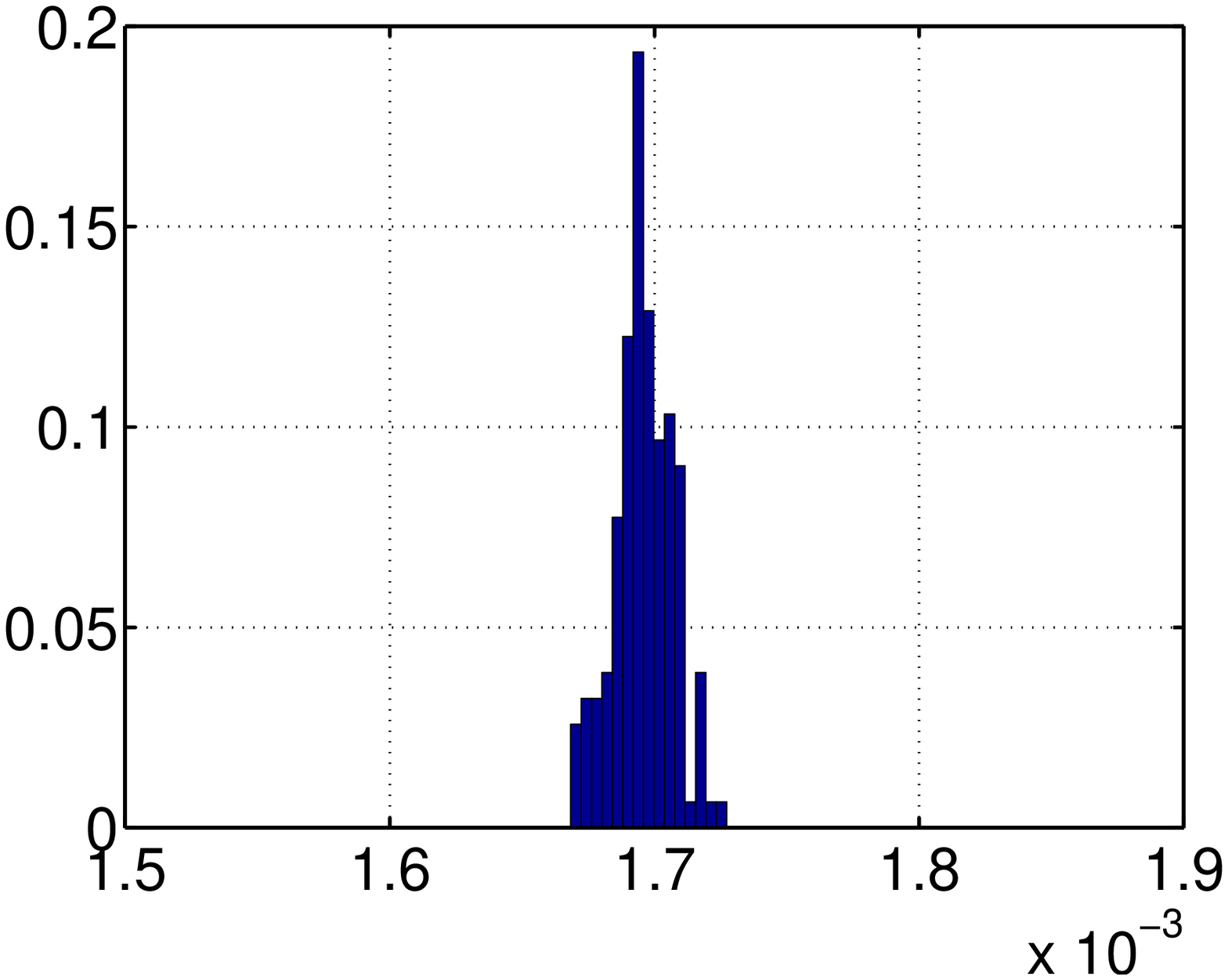} \label{fig:gXhist}}

%Figs.~\ref{fig:gNchain} and~\ref{fig:gNhist}
%Figs.~\ref{fig:gXchain} and~\ref{fig:gXhist}

  %\caption{Chains (left) and histograms (right) of hyperparameters
  %  $\Gn$ (top) and $\Gx$ (bottom).\label{fig:results}}
  \caption{Chains and histograms  of hyperparameters
    $\Gn$   and $\Gx$.\label{fig:results}}

\end{figure}

\begin{figure*}[htbp]
  \centering

  \subfigure[True]{\includegraphics[width=0.45\textwidth]{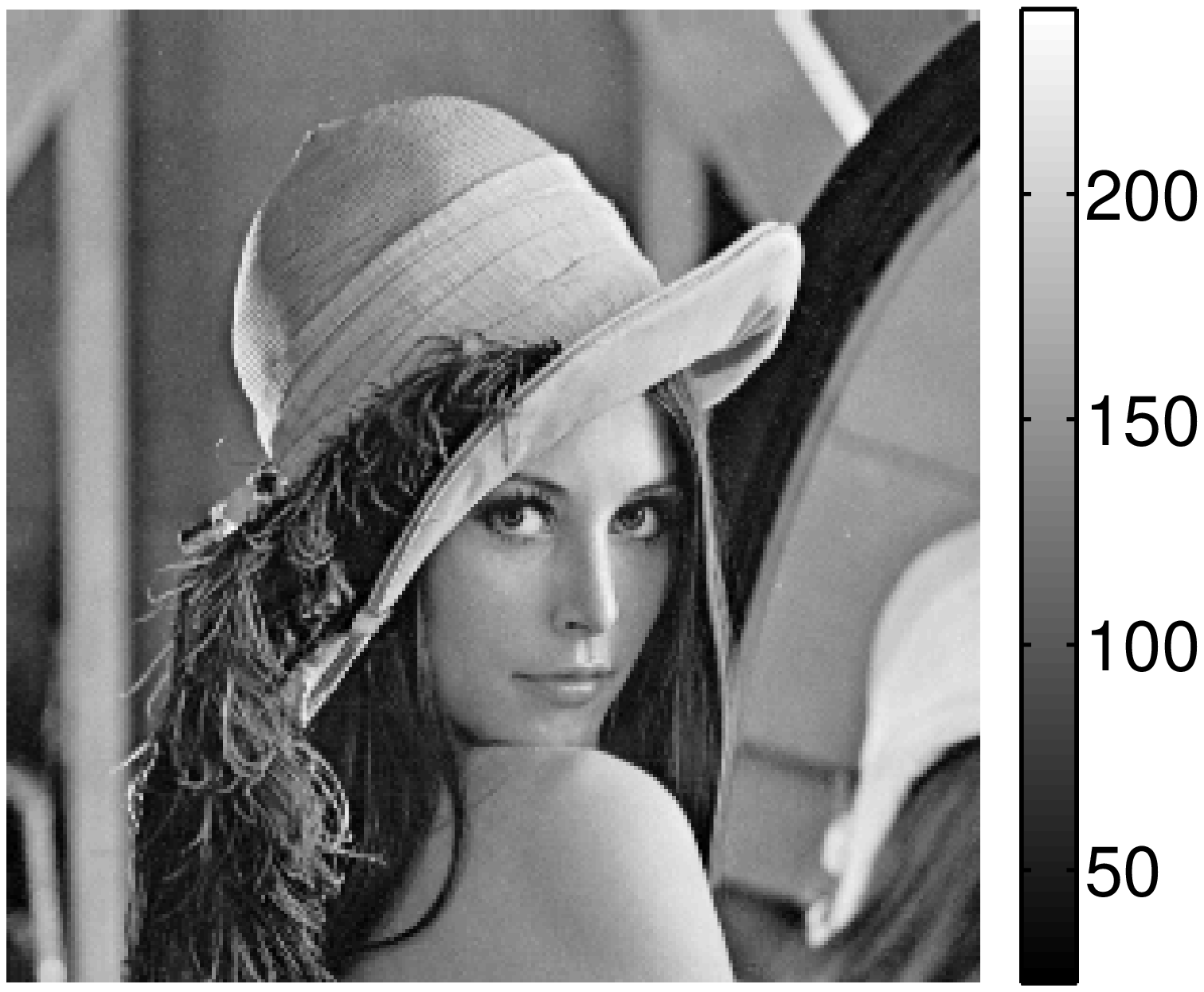} \label{fig:true}}%
  \subfigure[Data]{\includegraphics[width=0.45\textwidth]{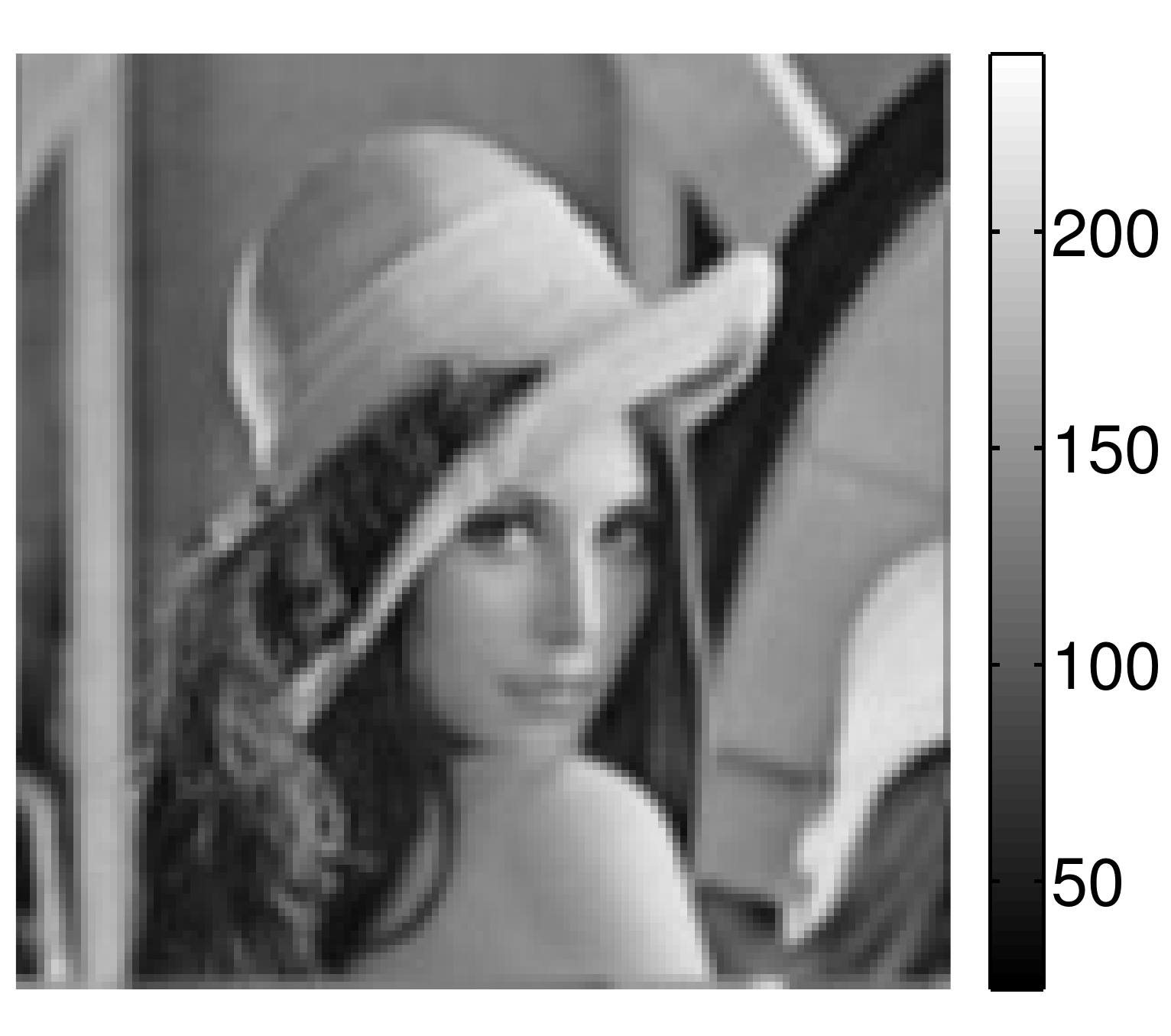} \label{fig:data}}

  \subfigure[Estimate]{\includegraphics[width=0.45\textwidth]{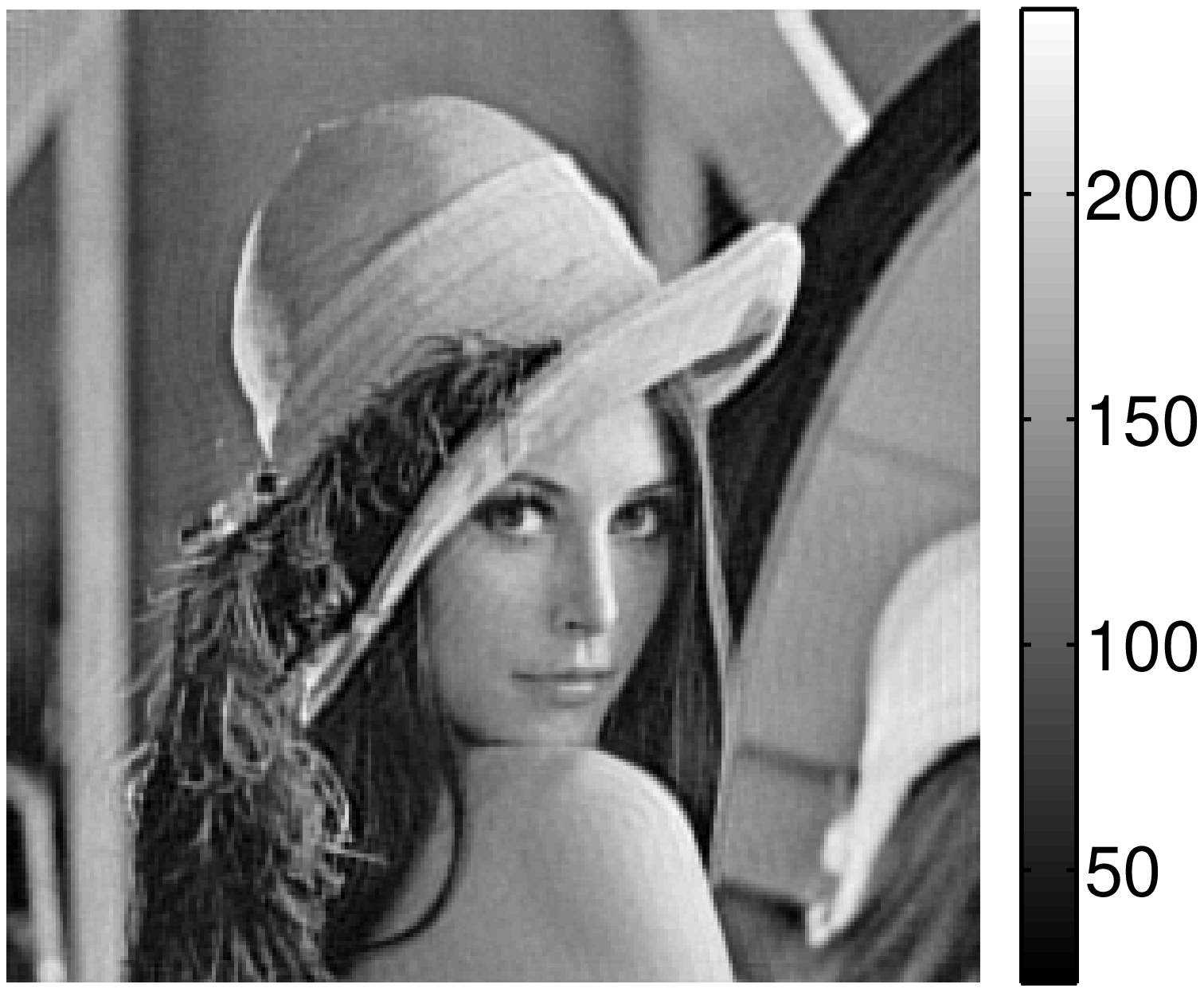} \label{fig:eap}}%
  \subfigure[Uncertainty]{\includegraphics[width=0.45\textwidth]{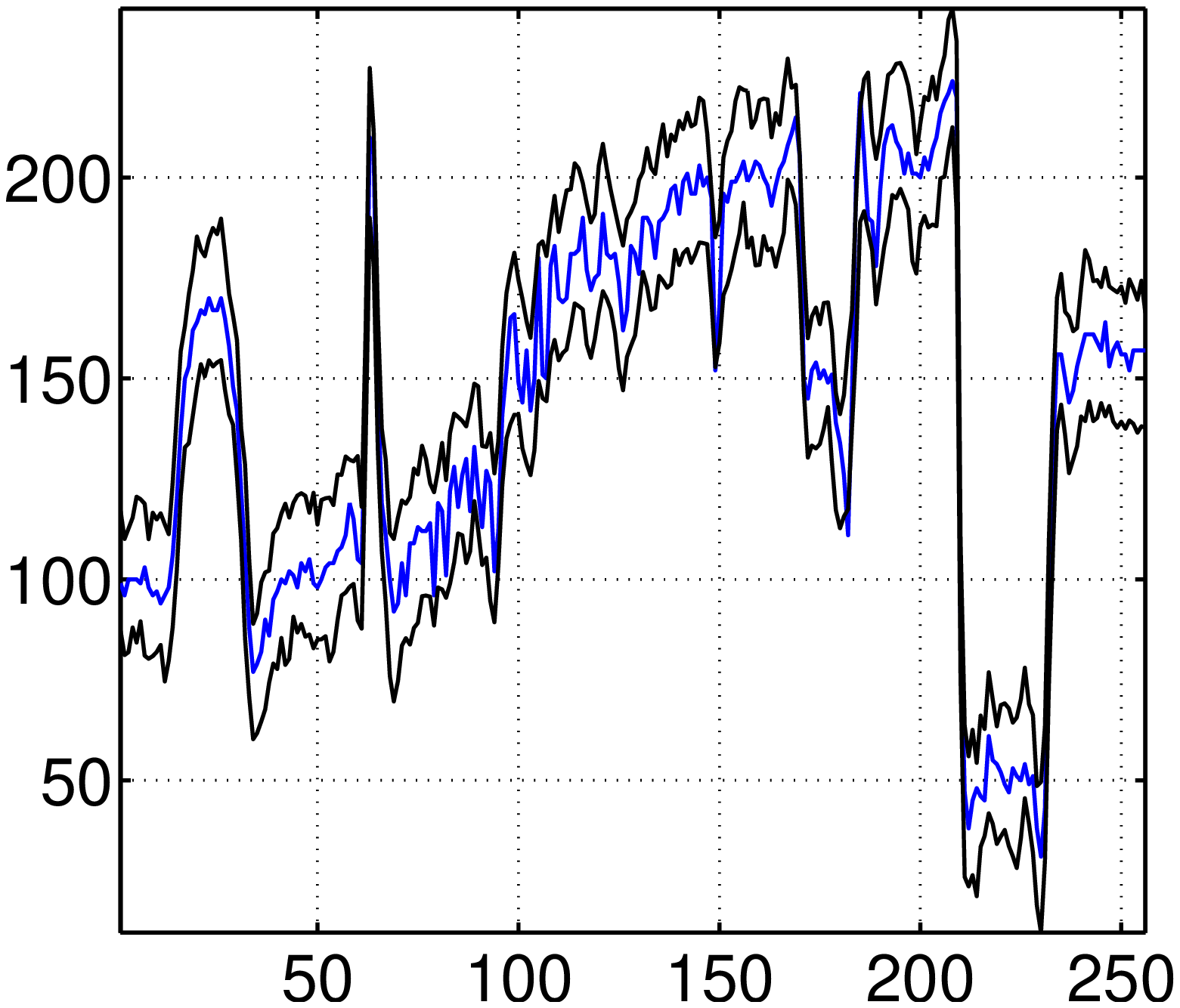} \label{fig:var}}

  \caption{Image reconstruction: true image~\ref{fig:true}, one of the
    low resolution images~\ref{fig:data} and the proposed
    estimate~\ref{fig:eap}. The plot \ref{fig:var} is a true image
    slice inside the 99\% confidence interval around the estimate.}

  \label{fig:imSR}
\end{figure*}

The usual forward model writes $\yb = \Hb\xb + \nb = \Pb \Cb\xb +
\nb$, where $\yb \in \eR^M$ collects the low resolution images (5
images of $128 \times 128$ pixels), $\xb \in \eR^N$ is the original
image ($256 \times 256$ pixels), $\nb$ is the noise, $\Cb$ and $\Pb$
are circulant convolution and decimation matrices. The prior density
for $\nb$ is $\mathcal{N}(\mathbf{0},\Gn^{-1}\Ib)$ and the one for
$\xb$ is $\mathcal{N}(\mathbf{0},\Gx^{ -1} \Db^t\Db)$ where $\Db$ is
the Laplacian operator. The hyperparameters $\Gn$ and $\Gx$ are
\textit{unknown} and their prior law are Jeffreys'.  The posterior
density is then
\begin{multline}
  \label{eq:1}
  p(\xb,\Gn,\Gx|\yb) ~~\propto~~ \Gn^{M/2 - 1} \Gx^{(N-1)/2 - 1} \\
  \Exp{-\frac{\Gn}{2}\|\yb -\Pb\Cb\xb\|^2 -\frac{\Gx}{2}\|\Db\xb\|^2}.
\end{multline}
It is explored by a Gibbs sampler:
iteratively sampling $\Gn$, $\Gx$ and $\xb$ under their respective
conditional probabilities
\begin{align*}
  p(\Gn^{(k)}|\xb,\Gx,\yb) & = \mathcal{G}\left(1 + M/2, 2/\left\|\yb
      - \Pb\Cb\xb^{(k-1)}\right\|^2\right) \\
  p(\Gx^{(k)}|\xb,\Gn,\yb) & = \mathcal{G}\left(1 + (N-1)/2,
    2/\left\|\Db \xb^{(k-1)}\right\|^2\right)  \\
  p(\xb^{(k)}|\Gx,\Gn,\yb) & = \mathcal{N}(\mxapost, \Rxapost)
\end{align*}
with
\begin{align*}
  \Rxapost & = \left(\Gn^{(k)} \Cb^t\Pb^t\Pb\Cb + \Gx^{(k)} \Db^t\Db
  \right)^{-1} \\
  \mxapost & = \Gn^{(k)}\Rxapost\Pb^t\Cb^t \yb.
\end{align*}
The conditional posteriors for the hyperparameters are Gamma laws and
consequently, easy to sample.

The conditional posterior for $\xb$ is Gaussian, but existing sampling
approaches are not operational due to the structure of the covariance
$\Rxapost$, as explained in Section~\ref{sec_prob_inv} with
$\Hb=\Pb\Cb$: $\Hb$ is non-circulant due to the decimation and $\Hb$
is not sparse especially in the case of large support. In this case,
the PO algorithm~2 directly provides a desired sample (with both
correct mean and correct covariance).

\begin{figure}[htbp]
  \centering
\end{figure}

It is important to keep in mind that the proposed PO algorithm does
not improve image quality itself (w.r.t. other SR methods) but the
crucial novelty is to allow for hyperparameter estimation. In this
sense, Fig.~\ref{fig:results} shows the hyperparameter iterates (that
illustrate the operation and convergence) and histograms (that
approximate marginal posteriors) ; the posterior means are $\widehat
\Gn \approx 8$ and $\widehat \Gx \approx 2\times 10^{-3}$. Concerning
the images themselves, results are shown in Fig.~\ref{fig:imSR}:
estimated image in~\ref{fig:eap} clearly shows a better resolution
than data in Fig.~\ref{fig:data} and it is visually close to the
original image of~\ref{fig:true}. It is then clear that the approach
produces correct hyperparameters i.e. correct balance between data and
prior. Moreover, uncertainties are derived from the samples through
the posterior standard deviation. It is illustrated in
Fig.~\ref{fig:var} which shows that the true image is inside the 99\%
confidence interval around the estimate. As a conclusion, the proposed
PO algorithm makes it possible to include sampling algorithms in SR
method whereas it was not possible before. It enables to provide joint
image and hyperparameters estimation as well as uncertainties
computations.

% The objective is to use MCMC to estimate $\xb$ but also the
% hyperparameters value, which as never been done in superresolution
% to the best of our knowledge

% The uncertainty about the hyperparameters are $2\times 10^{-3}$ and
% $1.5\times 10^{-5}$ for $\Gn$ and $\Gx$, respectively. They are
% relatively small thanks to the prior information and a favorable
% SNR.

% \GioSL{Ces valeurs de $\widehat \Gn \approx 8$ et $\Gn^{*} = 10$) ne
%   correspondent pas aux valeurs qui semblent apparaître sur la
%   figure\dots On est plutôt sur du 0.95 et 1\dots\Joli{Du reste, on
%     pourrait retirer ces deux phrases en vert caca d'oie concernant
%     la valeur des hyperparamètres\dots}}

\section{Conclusion}
\label{sec_conclu}

This paper presents a novel approach for sampling high-dimensional
Gaussian fields when usual approaches are ineffective. A sample of the
target density is produced as the minimizer of a precisely designed
quadratic criterion. It relies on a perturbation-optimization
principle: adequate stochastic perturbation of a criterion and
optimization of the perturbed criterion. It is shown that the
criterion minimizer is a sample of the target density. The approach is
applicable as soon as a particular factorization of the precision
matrix is available, and it is usually the case in inverse
problems. There is a wide class of applications, in particular any
data processing problem based on a linear forward model and
conditional Gaussian prior for noise and object. The effectiveness of
the proposed algorithm has been illustrated in \cite{orieux11,Feron07}
and in this paper on a more academic super-resolution imaging problem
allowing automatic tuning of hyperparameters.

% It relies on a perturbation-optimization principle: adequate
% stochastic perturbation of a criterion and optimization of the
% perturbed criterion. It is shown that the criterion minimizer is a
% sample of the \Gio{\KK{desired} target} density.

\section{Acknowledgment}
\label{sec:acknowledgment}

The authors would like to thank \JI (IRCyN) for inspiration of this
work \cite{IdierUnknown}, \TR and \AMD (L2S), for fruitful
discussions, and Cornelia \textsc{Vacar} (IMS) for carefully reading the paper.

\bibliographystyle{IEEEtran}

%\bibliography{biben,revuedef,revueabr,biblio}
\bibliography{biben,revuedef,revueabr,biblio}

\end{document}